\documentclass[11pt]{article}
\usepackage{amsfonts}
\usepackage{amsmath}
\usepackage{amssymb}
\usepackage{amsthm}
\usepackage{apalike}
\usepackage{graphicx}
\usepackage{hyperref}
\usepackage{setspace}
\usepackage[titletoc,title]{appendix}
\usepackage{a4wide}
\usepackage{color}
\usepackage{tikz, pgfplots}
\usepackage{tikz}
\usepackage{bera}
\usetikzlibrary{matrix}

\newtheorem{proposition}{Proposition}
\newtheorem{propositionA}{Proposition}

\theoremstyle{definition}
\newtheorem{remark}{Remark}
\newtheorem{remarkA}{Remark}

\newcommand{\E}{\mathbb{E}}

\begin{document}

\title{Black--Scholes--Merton Option Pricing Revisited:\\Did we Find a Fatal Flaw?\thanks{A first version of this paper appeared online on 9 January 2022. While any mistakes are ours, we thank without implication everyone who helped us by commenting on our analysis or with more general advice. We are especially grateful to Karma Dajani, as well as to Erik Balder, Alex Boer, Pedro Cazorla, Folkert van Cleef, John Cochrane, Timothy Crack, Freddy Delbaen, Ron Doney, Darrell Duffie, Richard Gill, Jakob de Haan, Michael Harrison, Jeroen van der Hoek, Neeltje van Horen, Julien Hugonnier, Jon Ingersoll, Ioannis Karatzas, David Kreps, Andreas Kyprianou, Roger Laeven, Iman van Lelyveld, Jussi Lindgren, Ronnie Loeffen, Roger Lord, Nick Nassuphis, Joe Pimbley, Rodney Ramcharan, Edo Schets, Alan Sokal, Rogier Swierstra, Babinu Uthup, Ton Vorst and Bernt \O ksendal, as well as participants in the discussion on \href{https://quant.stackexchange.com/questions/75683/three-mathematical-mistakes-in-black-scholes-merton-option-pricing}{Quant Stack Exchange}. Mark Mink (m.mink@dnb.nl) is an economist and works at De Nederlandsche Bank (DNB) in Amsterdam and Frans de Weert (frans@qbsm.nl) is a mathematician and former option trader and works at QBSM B.V in Zeist, both in The Netherlands. All views expressed are ours and not those of our employers.}}
\author{Mark Mink \and Frans J. de Weert}
\maketitle

\begin{center}
\large{Feedback or Suggestions?\\E-mail us or Join the Discussion on \href{https://quant.stackexchange.com/questions/75683/three-mathematical-mistakes-in-black-scholes-merton-option-pricing}{Quant Stack Exchange}}.
\end{center}
\hspace{1cm}

\begin{abstract}
The option pricing formula of Black and Scholes (1973) hinges on the continuous-time self-financing condition, which is a special case of the continuous-time budget equation of Merton (1971). The self-financing condition is believed to formalize the economic concept of portfolio rebalancing without inflows or outflows of external funds, but was never formally derived in continuous time. Moreover, and even more problematically, we discover a timing mistake in the model of Merton (1971) and show that his self-financing condition is misspecified both in discrete and continuous time. Our results invalidate seminal contributions to the literature, including the budget equation of Merton (1971), the option pricing formula of Black and Scholes (1973), the continuous trading model of Harrison and Pliska (1981), and the binomial option pricing model of Cox, Ross and Rubinstein (1979). We also show that Black and Scholes (1973) and alternative derivations of their formula implicitly assumed the replication result.
\end{abstract}

\clearpage

\section{Introduction}
The option pricing formula of Black and Scholes (1973)\nocite{blackscholes1973}, among many seminal results in the continuous-time finance literature, hinges on the continuous-time \textit{self-financing} condition.\footnote{Harrison and Kreps (1979)\nocite{harrisonkreps1979} coined the term self-financing, and point out that ``this restriction is implicit in the original treatment of Black and Scholes (1973) and is explicitly displayed in Merton (1977)\nocite{merton1977}."} This condition is believed to formalize the economic concept of an asset portfolio that is rebalanced without inflows or outflows of external funds. Remarkably, however, the continuous-time self-financing condition has never been derived in continuous time. Merton (1971)\nocite{1971} states that ``it is not obvious" how to do this, so that ``it is necessary to examine the discrete-time formulation of the model and then to take limits carefully to obtain the continuous-time form." This discrete-time analogy has remained the sole motivation for the continuous-time specification of the condition, although Harrison and Pliska (1981)\nocite{harrisonpliska1981} once cautiously stated to ``have no doubt that these are the right definitions, but a careful study of this issue is certainly needed." Bj\"{o}rk (2009)\nocite{bjork2009} emphasizes the informal ``motivating nature" of the discrete-time analogy, but also in this modest role the analogy is not straightforward. The reason is that the path-by-path interpretation of the discrete-time model is lost in the continuous-time model, since the It\^{o} integrals used in continuous time hold almost surely but not pathwise (Bick and Willinger, 1994).

The absence of a formal proof for the continuous-time self-financing condition is a bit bewildering, since the condition played a key role in making techniques from continuous-time stochastic calculus available for the analysis of intertemporal consumption and portfolio optimization (e.g., Merton 1971, 1973a)\nocite{merton1971}\nocite{merton1973ICAPM}.\footnote{This analytical framework became more technically sophisticated over time, but still hinges on the self-financing condition. See, for example, textbooks by Bj\"{o}rk (2009)\nocite{bjork2009}, Delbaen and Schachermayer (2006)\nocite{delbaenschachermayer2006}, Duffie (2001)\nocite{duffie2001}, Karatzas and Kardaras (2021)\nocite{karatzaskardaras2021}, Lamberton and Lapeyre (2011)\nocite{lambertonlapeyre2011}, Merton (1990)\nocite{merton1990} and Shreve (2004)\nocite{shreve2004}.} From a mathematical perspective, to be able to use these techniques it does not suffice to define the condition based on a discrete-time analogy. Instead, one needs to prove in continuous time that the continuous-time condition is correctly specified. Our analysis shows, however, that the problem with the self-financing condition is much bigger than this. Specifically, we show that the timing in the model of Merton (1971) does not adequately reflect the sequencing of rebalancing in his economic narrative. As a result, both his discrete-time and his continuous-time model inadvertently assume that the investor knows which stochastic process drives the stock return. However, even when agents in the model know that the stock price in continuous time is driven by a Wiener process, for example, they cannot know which one out of the infinitely many potential Wiener processes this is (just as they do not have a coin that always comes up heads when the stock goes up and tails when it goes down). In fact, because mistakes like these are so easily made, it is good practice among mathematicians to label different Wiener processes, as this reduces the risk of accidentally equating one to the other. By establishing that both the discrete-time and the continuous-time self-financing condition are misspecified, our analysis also invalidates the continuous-time budget equation of Merton (1971), the option pricing formula of Black and Scholes (1973), and the continuous trading model of Harrison and Pliska (1981, 1983)\nocite{harrisonpliska1983}.

At this point, we may have left some readers in disbelief or even anger about such a bold statement as the one above. After all, as Warren Buffet observed in 2009, the option pricing formula of Black and Scholes (1973) ``has approached the status of holy writ in finance." Such a status, of course, also provides fertile ground for confirmation bias. Moreover, the rapid growth of the literature may have obfuscated the fact that a formal proof for the continuous-time self-financing condition was never provided. Academics and practitioners could therefore have come to believe quite naturally that the self-financing condition had a stronger substantiation than actually was the case. With our new results in mind, however, we look back with some amazement at earlier indications in the literature that something was wrong with the formula, and the subsequent attempts to rationalize these issues without questioning the formula itself.

Firstly, the paper of Black and Scholes (1973) was known to contain mathematical mistakes. These mistakes were considered to have been repaired by Merton (1973) with the self-financing condition (e.g., Bergman (1981)\nocite{bergman1981}, Beck (1993)\nocite{beck1993}, Bartels (1995)\nocite{bartels1995}, MacDonald (1997)\nocite{macdonald1997}, and Carr (1999)\nocite{carr1999}), even though he did not prove this condition in continuous time.\footnote{Long (1974)\nocite{long1974} points out that the replication result hinges on the assumption by Black and Scholes (1973) that the option price is only a function of the stock price and the residual maturity of the option, and not of other time-varying factors such as the market price of risk.} Secondly, F\"{o}llmer (2001)\nocite{follmer2001} showed that the definition of a self-financing trading strategy implies an arbitrage opportunity if the stock price process has zero quadratic variation. Most studies therefore a priori exclude these stock price processes (such as the fractional Brownian motion, see also Bj\"{o}rk and Hult, 2005\nocite{bjorkhult2005}), even though there is no economic reason why these processes would by definition imply an arbitrage opportunity in an efficient market. Thirdly, the theoretical result of Black and Scholes (1973) was contradicted by several empirical findings, such as the implied volatility surface (Macbeth and Merville, 1979\nocite{macbethmerville1979}), the imperfect correlation between call prices and the underlying stock (Bakshi, Cao and Chen, 2000\nocite{bakshicaochen2000}), the excess returns on delta-hedged positions (Bakshi and Kapadia, 2003\nocite{bakshikapadia2003}), and the credit spread puzzle implied by the Merton (1974)\nocite{merton1974} model (Jones, Mason and Rosenfeld, 1984\nocite{jonesmasonrosenfeld1984}). These findings were interpreted as evidence against the idealized capital market conditions that Black and Scholes (1973) assumed, but their formula itself remained unquestioned, although it hinged on a condition that lacked a formal proof.

While our result that the self-financing condition is misspecified offers a concise explanation for the above phenomena, the feedback we received on previous versions of this analysis helped us to pinpoint two additional mathematical mistakes in the framework of Black and Scholes (1973). Firstly, even if the self-financing condition would have been correct, we point out that the alternative proofs by Merton (1973b, 1977) and Harisson and Pliska (1981, 1983), which were believed to have repaired the mistakes in the original analysis of Black and Scholes (1973), assume the replication result that is to be established. As Mink and de Weert (2024)\nocite{minkdeweert2024} discuss in larger detail, the claim that an option can be replicated with a stock and a risk free bond was therefore an assumed result without a formal proof. Secondly, even if one also sidesteps this problem with their analytical framework, we show that the partial differential equation of Black and Scholes (1973) implies that there are paths where the rebalanced portfolio is not self-financing or does not replicate the option. Hence, we conclude that for multiple reasons, the option pricing framework of Black and Scholes (1973) is fatally flawed.

To provide more context to the result that the option pricing formula of Black and Scholes (1973) is incorrect, we also discuss the statement of Cox, Ross and Rubinstein (1979)\nocite{coxrossrubinstein1979} that the continuous-time formula of Black and Scholes (1973) is equal to the limit of their discrete-time binomial option pricing formula. Our analysis invalidates the discrete-time hedging argument, however, by establishing that in discrete time the self-financing condition is incorrect as well. Moreover, a closer look at the binomial model shows that Cox, Ross and Rubinstein (1979) implicitly assume that options and their underlying stocks are only exposed to systematic risk (i.e., to the market portfolio) and not to idiosyncratic risk. We then show that this assumption, which is difficult to rationalize from an economic perspective, directly implies their difference equation, without using the discrete-time self-financing condition or a hedging argument. From a mathematical perspective, the difference equation that underlies the binomial option pricing formula of Cox, Ross and Rubinstein (1979) therefore reflects an implicit assumption rather than a hedging argument. Hence, the hedging argument of Black and Scholes (1973) cannot be interpreted as the limit of a discrete-time hedging argument by Cox, Ross and Rubinstein (1979), since there is no discrete-time hedging argument to begin with.

The implications of our analysis for academic research are difficult to oversee, since the continuous-time and discrete-time self-financing conditions underlie many seminal results in the finance literature. For financial market participants, furthermore, the absence of a correct option pricing formula implies that one can only speculate about the sign and magnitude of the mispricing that results when using the formula of Black and \text{Scholes} (1973). Moreover, there is no longer a basis for the argument that hedging an option with the underlying stock yields a risk-free portfolio. While option traders do not use the formula of Black and \text{Scholes} (1973) directly, they typically use the underlying hedging argument. This argument allowed traders to price options solely based on their view of the volatility of the underlying stock, without the need to have a view on the expected return of the stock and the expected return of the option. By disproving this hedging argument, our analysis implies that options are not redundant securities and that trading them cannot be reduced to trading the volatility of the underlying stock. Still, the value of an option can be expressed as its expected payoff at maturity discounted by its expected return. This expression resembles the option pricing formula of Black and Scholes (1973), but contains the percentage drift on the stock and the (time-varying) percentage drift on the option instead of the risk-free rate of return.

\section{Analysis}\label{analysis}
We first summarize a standard derivation of the replicating portfolio that underlies the option pricing formula of Black and Scholes (1973)\nocite{blackscholes1973}. Staying close to their original notation, consider a call option with value $w_{t}$ at time $t$ that has a strike price $k$ and a maturity date $T$ and is written on a stock with value $x_{t}$. In continuous time and conditional on $\mathcal{F}_{t_{0}}$, the value of the stock follows a geometric Brownian motion:
\begin{align}\label{dx}
\int_{t_{0}}^{t_{1}}dx_{t}=\int_{t_{0}}^{t_{1}}x_{t}\mu dt+\int_{t_{0}}^{t_{1}}x_{t}\sigma dW_{t}\hspace{0.25cm}\text{with $x_{t_{0}}>0$ given},
\end{align}
where $\mu$ is the percentage drift, $\sigma>0$ is the percentage standard deviation, and $\{W_{t}\}_{t\geq t_{0}}$ is a Wiener process on a filtered probability space $(\Omega, \mathcal{F}, \{\mathcal{F}_t\}_{t\geq t_{0}}, \mathbb{P})$ satisfying the usual conditions. We use integral notation $\int dx_{t}$ instead of the informal $dx_{t}$ to avoid confusion, since equation (\ref{dx}) has no path-by-path interpretation due to the unbounded total variation of the Wiener process (e.g., $dx_{t}$ is not a ``small change" in the stock price, a point to which we return later). The distinction between bounded and unbounded variation is sometimes overlooked when interpreting stochastic integrals, but we assume that the reader is familiar with these mathematical concepts. Black and Scholes derive a pricing formula for the option by constructing a replicating portfolio, but to date it is well-known that their analysis of this portfolio contains mathematical mistakes (see also Mink and de Weert, 2024\nocite{minkdeweert2024}). Their portfolio is therefore typically replaced by the one defined by Merton (1973b, 1977)\nocite{merton1973}\nocite{merton1977}, who replicates the option with a portfolio of $w_{1,t}$ stocks and $\beta_t$ risk-free bonds with value $b_t$ and return $db_{t}= b_{t}rdt$. Specifically, Merton (1973b, 1977) defines the option value as: 
\begin{equation}\label{wval}
w_t=w_{1,t}x_t+\beta_t b_t,
\end{equation}
where $w_{1,t}=\partial w_{t}/\partial x_{t}$ is the partial derivative of the option value with respect to the value of the stock (i.e., delta). This partial derivative changes over time as it depends on the value of the stock and on the remaining maturity of the option. Black and Scholes thus assume that the number of stocks changes over time as well, and refer to this trading strategy as \textit{continuous rebalancing}. Applying the product rule of stochastic integration to equation (\ref{wval}) shows that the return on the option conditional on $\mathcal{F}_{t_{0}}$ equals:
\begin{equation}\label{prsi}
\int_{t_{0}}^{t_{1}}dw_t=\int_{t_{0}}^{t_{1}}w_{1,t}dx_t+\int_{t_{0}}^{t_{1}}dw_{1,t}x_{t}+\int_{t_{0}}^{t_{1}}dw_{1,t}dx_{t}+\int_{t_{0}}^{t_{1}}\beta_{t}db_{t}+ \int_{t_{0}}^{t_{1}}d\beta_tb_{t}+ \int_{t_{0}}^{t_{1}}d\beta_tdb_{t}.
\end{equation}
Merton (1973b, 1977) assumes that the next equality also holds:
\begin{equation}\label{ctsfc}
\int_{t_{0}}^{t_{1}}dw_{1,t}x_{t}+ \int_{t_{0}}^{t_{1}}dw_{1,t}dx_{t}+\int_{t_{0}}^{t_{1}}d\beta_t b_{t}+ \int_{t_{0}}^{t_{1}}d\beta_tdb_{t}=0,
\end{equation}
which he substitutes into equation (\ref{prsi}) to obtain the return on the option:
\begin{equation}\label{dw}
\int_{t_{0}}^{t_{1}}dw_t=\int_{t_{0}}^{t_{1}}w_{1,t}dx_t+\int_{t_{0}}^{t_{1}}\beta_t db_t.
\end{equation}
In summary, the replicating portfolio of Black and Scholes is described by equations (\ref{wval}) and (\ref{dw}), and hinges on the validity of equation (\ref{ctsfc}). This equation was named the \textit{continuous-time self-financing condition} by Harrison and Kreps (1979)\nocite{harrisonkreps1979}, who point out that ``this restriction is implicit in the original treatment of Black and Scholes (1973) and is explicitly displayed in Merton (1977)\nocite{merton1977}." Indeed, by using the self-financing condition, Merton (1973b, 1977) works around a mathematical mistake in the original derivation of Black and Scholes, who calculate the return on their replicating portfolio without applying the product rule of stochastic integration (e.g., Bergman, 1981).

While it is trivially obvious that equation (\ref{ctsfc}) holds for an asset portfolio that is not rebalanced (i.e., for $dw_{1,t}=d\beta_{t}=0)$, the self-financing condition is also believed to hold for asset portfolios that are continuously rebalanced without requiring an inflow or outflow of external funds, based on an analysis of Merton (1971)\nocite{merton1971}. In this analysis, Merton (1971) introduces the condition as a backward-looking accounting relationship in discrete time, after which he derives a forward-looking version of this relationship and takes the limit to obtain equation (\ref{ctsfc}) in continuous time.\footnote{Specifically, the continuous-time self-financing condition is equal to equation (9') in Merton (1971)\nocite{merton1971} when consumption $C_{t}dt=0$, in which case the rebalanced portfolio does not exhibit inflows or outflows of funds.\label{budget}} We show below, however, that the self-financing condition of Merton (1971) is misspecified because of mathematical mistakes in his analysis. Before we discuss these mistakes in more detail, the next proposition first establishes the correct specification of the continuous-time self-financing condition. The proof for the proposition uses the same limiting procedure that Merton (1971) uses to justify equation (\ref{ctsfc}), and shows that this procedure yields an additional constraint that the literature has overlooked. The proof applies this limiting procedure to the discrete-time self-financing condition of Harrison and Pliska (1981), although Proposition \ref{correctsfcmerton} in the appendix shows that using the discrete-time condition of Merton (1971) yields the same result. Still, it is remarkable that two seminal option pricing papers disagree about the correct specification of the self-financing condition in discrete time, and use the same (but incorrect) specification in continuous time.

\begin{proposition}\label{correctsfc}
For a rebalanced portfolio of $\alpha_{t}$ stocks and $\beta_{t}$ risk-free bonds, the correct continuous-time self-financing condition is $\int_{t_{0}}^{t_{1}}d\alpha_{t}x_{t}+\int_{t_{0}}^{t_{1}}d\beta_t b_{t}=0$, and the (path-by-path defined) processes $\alpha_{t}$ and $\beta_{t}$ must be continuous with bounded variation.
\end{proposition}
\begin{proof}
Using our notation, and conditional on $\mathcal{F}_{t_{0}}$, Harrison and Pliska (1981) state in their equation (2.2) that a trading strategy is self-financing in discrete time if:
\begin{equation}
\alpha_{t}x_{t}+\beta_{t}b_{t}=\alpha_{t+\Delta t}x_{t}+\beta_{t+\Delta t}b_{t},
\end{equation}
for all $t_{0}\leq t<T$, where $T$ is the maturity date of the option. Rearranging this equation yields:
\begin{equation}\label{discretesfc}
\Delta \alpha_{t}x_{t}+\Delta \beta_{t}b_{t}=0,
\end{equation}
where we defined $\Delta \alpha_{t}=\alpha_{t+\Delta t}-\alpha_{t}$ and $\Delta \beta_{t}=\beta_{t+\Delta t}-\beta_{t}$. Note that equation (\ref{discretesfc}) must hold path-by-path, otherwise there would be an inflow or outflow of funds for some realizations of $\Delta \alpha_{t}$ and $\Delta \beta_{t}$. When defining $\Delta t=\left(T-t\right)/n$ with $n>0$, the corresponding equation when the investor rebalances continuously equals:
\begin{equation}\label{limsfcexante}
\lim_{n \to \infty}\left[\Delta \alpha_{t}x_{t}+\Delta\beta_{t} b_{t}\right]=0,
\end{equation}
which in stochastic integral notation gives the continuous-time self-financing condition:
\begin{equation}
\int_{t_{0}}^{t_{1}}d\alpha_{t}x_{t}+ \int_{t_{0}}^{t_{1}}d\beta_t b_{t}=0,
\end{equation}
where the processes $\alpha_{t}$ and $\beta_{t}$ must be continuous with bounded variation, because equation (\ref{limsfcexante}) is defined path-by-path. Because this bounded variation implies that $\int_{t_{0}}^{t_{1}}d\alpha_{t}dx_{t}=0$ and $\int_{t_{0}}^{t_{1}}d\beta_{t}db_{t}=0$, the return on the portfolio then equals $\int_{t_{0}}^{t_{1}}\alpha_{t}dx_{t}+ \int_{t_{0}}^{t_{1}}\beta_t db_{t}$.
\end{proof}

The previous proposition illustrates that when taking the limit of the discrete-time model, one needs to impose the constraint that the processes $\alpha_{t}$ and $\beta_{t}$ have bounded variation in continuous time (just as in discrete time). Merton (1971) overlooks this constraint, and as a result runs into the issue of whether the terms $\int_{t_{0}}^{t_{1}}d\alpha_{t}dx_{t}$ and $\int_{t_{0}}^{t_{1}}d\beta_{t}db_{t}$ should be interpreted as portfolio returns or as cash flows from rebalancing (see his footnote 10). When the processes $\alpha_{t}$ and $\beta_{t}$ have bounded variation, however, both terms are equal to zero so that this question naturally disappears. In addition, the constraint that the processes $\alpha_{t}$ and $\beta_{t}$ have bounded variation ensures that the investor can determine how his stock and bond position changed after he rebalanced his portfolio. Bick and Willinger (1994) point out that this is not possible for the replication strategy of Black and Scholes, who set $\alpha_{t}=w_{1,t}$ even though the process $w_{1,t}$ has unbounded variation (and $\int_{t_{0}}^{t_{1}}dw_{1,t}dx_{t}=\int_{t_{0}}^{t_{1}}w_{11,t}\sigma^{2}x_{t}^{2}dt$ is not equal to zero). An investor who tries to implement their strategy therefore cannot determine how the number of stocks in his portfolio changed over time and how many stocks he currently holds (as Remark \ref{r:pathwise} in the appendix illustrates). Just as the previous issue, this issue disappears when the processes $\alpha_{t}$ and $\beta_{t}$ have bounded variation, but this constraint also implies that setting $\alpha_{t}=w_{1,t}$ is not a valid choice for a self-financing trading strategy.

While Proposition \ref{correctsfc} implies that the self-financing condition of Merton (1971) is misspecified in continuous time, a timing mistake in his model causes him to arrive at an incorrect discrete-time condition as well (the above proof for Proposition \ref{correctsfc} therefore used the discrete-time self-financing condition of Harrison and Pliska (1981), which does not suffer from this issue). The next remark shows that this second mistake originates at a point where his (appealing) economic narrative is not adequately formalized by his mathematical model. As a result, the discrete-time self-financing condition of Merton (1971) inadvertently assumes that the behaviour of the investor is (contemporaneously) driven by the same random number generator as the one that drives the stock returns. However, while the behaviour of the investor can be assumed to have the same distribution as the stock price movements, one cannot assume that the investor can generate the same random outcomes.\footnote{Even if the stock price process would have bounded variation, this is another reason why setting $\alpha_{t}=w_{1,t}$ is not a valid choice for a trading strategy, since $w_{1,t}$ is driven by the same Wiener process as $x_{t}$.} Because mistakes like these are so easily made, it is good practice among mathematicians to label different random number generators (e.g., Wiener processes), as this reduces the risk of accidentally equating one to the other.

\begin{remark}\label{r:eaep}
For a portfolio of $\alpha_{t}$ stocks and $\beta_{t}$ risk-free bonds, Merton (1971) models discrete-time rebalancing without inflows or outflows of funds with the ex-post condition:
\begin{equation}\label{dtsfcexpost}
[\Delta \alpha_{t} x_{t}+\Delta \alpha_{t} \Delta x_{t}+\Delta\beta_{t} b_{t}+\Delta\beta_{t} \Delta b_{t}\vert\mathcal{F}_{t+\Delta t}]=0.
\end{equation}
While Merton (1971) emphasizes in his economic narrative that the investor first observes the stock return $\Delta x_{t}$ and thereafter chooses $\Delta \alpha_{t}$ and $\Delta \beta_{t}$, we note that his model does not reflect this sequencing since all variables are given the same time subscript $t$.\footnote{Merton explains in his narrative that the change in the number of stocks $[\Delta \alpha_{t}\vert\mathcal{F}_{t+\Delta t}]$ is implemented just before $t+\Delta t$, during an infinitesimal ``period" after the realization of the stock return $[\Delta x_{t}\vert\mathcal{F}_{t+\Delta t}]$. When one would try to formalize this narrative mathematically, however, one runs into more fundamental mathematical questions such as that a period shorter than $\Delta t$ does not exist in a discrete-time model.} While this may seem a minor issue, the implications thereof become visible when Merton (1971) argues that equation (\ref{dtsfcexpost}) also holds stochastically (i.e., ex-ante), so that:
\begin{equation}\label{dtsfcexante}
[\Delta \alpha_{t} x_{t}+\Delta \alpha_{t} \Delta x_{t}+\Delta\beta_{t} b_{t}+\Delta\beta_{t} \Delta b_{t}\vert\mathcal{F}_{t}]=0.
\end{equation}
Unless $[\Delta \alpha_{t}\vert\mathcal{F}_{t}]=0$, however, this equation implies that an investor can ex-ante choose the stochastic processes $[\Delta \alpha_{t}\vert\mathcal{F}_{t}]$ and $[\Delta \beta_{t}\vert\mathcal{F}_{t}]$ so that they (contemporaneously) depend on the stochastic process for the stock returns $[\Delta x_{t}\vert\mathcal{F}_{t}]$. Even if the investor knows the distribution of these future returns, however, he does not have the underlying random number generator. For example, in the context of continuous-time equation (\ref{dx}), the investor knows that the stock returns are driven by a Wiener process, but he does not know which one out of the infinitely many Wiener processes this is (e.g., he cannot distinguish whether the returns are driven by $\{W_{t}\}_{t\geq t_{0}}$ or by $\{W_{t}^{*}\}_{t\geq t_{0}}$, or by $\{W_{t}^{**}\}_{t\geq t_{0}}$ etc.). The Wiener process that drives his behaviour (i.e., how he changes the number of stocks and bonds in his portfolio) therefore cannot be equal to the one that drives the stock returns. The discrete-time self-financing condition in equation (\ref{dtsfcexante}) is therefore misspecified, and should be defined in line with Harrison and Pliska (1981) as equation (\ref{discretesfc}) instead.\footnote{Equation (\ref{discretesfc}) can also be derived by repairing the timing mistake in the model of Merton (1971), who argues that the numbers of stocks and bonds in equation (\ref{dtsfcexpost}) are changed after observing $x_{t+\Delta t}$. These changes therefore cannot have time subscript $t$ but should have the next time subscript $t+\Delta t$. Writing $x_{t+\Delta t}=x_{t}+\Delta x_{t}$, equation (\ref{dtsfcexpost}) then becomes $\left[\Delta \alpha_{t+\Delta t}x_{t+\Delta t}+\Delta\beta_{t+\Delta t} b_{t+\Delta t}\vert \mathcal{F}_{t+2\Delta t}\right]=0$, which for one period earlier equals $\left[\Delta \alpha_{t}x_{t}+\Delta\beta_{t} b_{t}\vert \mathcal{F}_{t+\Delta t}\right]=0$ and therefore yields equation (\ref{discretesfc}) when conditioning on $\mathcal{F}_{t}$.}
\end{remark}

In summary, the above analysis shows that the continuous-time and the discrete-time self-financing condition of Merton (1971) are both misspecified. While Harrison and Pliska (1981) adopt the same continuous-time condition as Merton (1971), their discrete-time is different and is correctly specified. Using this discrete-time specification to derive a new and correct continuous-time self-financing condition shows that the latter has the economically appealing feature of being interpretable path-by-path, just as its discrete-time counterpart. At the same time, the correctly specified continuous-time self-financing condition implies that the replicating portfolio does not yield a valid rebalancing strategy, which is the \underline{first fatal flaw} that we establish in the framework of Black and Scholes.

The feedback that we received on previous versions of this analysis helped us to pinpoint two additional flaws in the framework of Black and Scholes that can be regarded as fatal in themselves. For the sake of argument, we therefore set aside the problems with the continuous-time self-financing condition and focus on the second mathematical issue. More specifically, we consider the claim of Black and Scholes that an investor can replicate an option by rebalancing a portfolio of stocks and risk-free bonds. Of course, one can take such a rebalanced portfolio of stocks and risk-free bonds as starting point for the analysis, but the claim that such a portfolio can replicate an option then requires a formal proof (in a discrete-time context, this claim is known to be incorrect if the stock price has more than two potential outcomes at the end of each period, as is discussed in Section \ref{binomial} on the binomial model of Cox, Ross and Rubinstein, 1979). Black and Scholes try to prove this claim by establishing that a long position of one option plus a continuously rebalanced short-position of $w_{1,t}$ stocks yields a risk-free bond. However, as their analysis contains several mathematical mistakes (as summarized in Mink and de Weert, 2024\nocite{minkdeweert2024}), their proof is not used in the literature anymore.

While Black and Scholes do not provide a (correct) proof for their claim that an option can be replicated with a rebalanced portfolio of stocks and risk-free bonds, the literature generally refers to Harrison and Pliska (1981, 1983)\nocite{harrisonpliska1983} for a rigorous derivation of this result. These authors build a continuous trading model in which every option can be ``attained" with a rebalanced portfolio of securities, if the prices of these securities satisfy a certain martingale representation property. Crucially, however, while a replicating strategy is also an attaining strategy, the converse is not necessarily true. The reason is that a replicating strategy has the same value as the option for all $t\leq T$, while an attaining strategy has the same value as the option at the maturity date $T$, and may or may not have the same value as an option at $t<T$. As Harrison and Pliska (1981) point out, there can thus be many alternative trading strategies that all attain the same option, and these strategies may not all have the same price. While the rebalanced portfolio of Black and Scholes is just one example of an attaining strategy, Harrison and Pliska (1981, 1983) discard all other strategies (i.e., these are not defined ``admissible") to ensure that this strategy becomes unique in their model. While their model is therefore consistent with the claim that the trading strategy of Black and Scholes replicates an option, it does not prove that this claim is correct.

In summary, both the original analysis of Black and Scholes as well as the subsequent analysis of Harrison and Pliska (1981, 1983) do not prove that the option can be replicated by a rebalanced portfolio of stocks and risk-free bonds. Likewise, and as discussed in more detail by Mink and de Weert (2024)\nocite{minkdeweert2024}, Merton (1973b, 1977) does not prove this replication result either. That is, using the terminology of Harrison and Pliska (1981, 1983), he analyses one specific trading strategy that attains the option, but disregards any other ones that may have a different price, and does not prove that his strategy not only attains but also replicates the option.\footnote{In addition, both Merton (1973b, 1977) and Harrison and Pliska (1981, 1983) assume that the value of a rebalanced portfolio cannot become negative. Without proof this is a strong claim from a stochastic perspective, because the rebalanced portfolio contains a short bond position. In discrete time the portfolio value can therefore become negative when the stock price drops by a sufficiently large amount. One could argue that this claim is nevertheless justified in continuous time, because the investor can avoid this scenario by continuously rebalancing the portfolio. Such an argument should then be established mathematically, however, and cannot simply be stated without proof.} We therefore conclude that the replication result is implicitly assumed rather than mathematically derived, which is the \underline{second fatal flaw} in the framework of Black and Scholes.

Again for the sake of argument, we now set aside both flaws discussed above and just assume that the rebalanced portfolio is self-financing and also replicates an option. The next proposition shows, however, that these assumptions are then contradicted by the option pricing formula of Black and Scholes. More specifically, the partial differential equation associated with the option pricing formula implies that there are paths where the rebalanced portfolio does not satisfy the self-financing condition or does not replicate the option (likewise, Proposition \ref{p:rp} in the appendix shows that the option pricing formula implies that there are paths where the self-financing condition does not hold). This contradiction is therefore the \underline{third fatal flaw} in the analytical framework of Black and Scholes.

\begin{proposition}
The partial differential equation of Black and Scholes implies that there are paths where the rebalanced portfolio does not satisfy the self-financing condition or does not replicate the option.
\end{proposition}
\begin{proof}
Black and Scholes define the value of their rebalanced portfolio as:
\begin{equation}\label{ptf}
w_{t}=w_{1,t}x_{t}+\beta_{t}b_{t},
\end{equation}
where $b_{t}$ is a risk-free bond. Applying the product rule of stochastic integration yields:
\begin{equation}
\int_{t_{0}}^{t_{1}}dw_{t}=\int_{t_{0}}^{t_{1}}w_{1,t}dx_{t}+\int_{t_{0}}^{t_{1}}\beta_tdb_{t}+\int_{t_{0}}^{t_{1}}dw_{1,t}x_t+\int_{t_{0}}^{t_{1}}dw_{1,t}dx_{t}+\int_{t_{0}}^{t_{1}}d\beta_t b_t+\int_{t_{0}}^{t_{1}}d\beta_t db_t,
\end{equation}
which Black and Scholes combine with the continuous-time self-financing condition:
\begin{equation}\label{sfc}
\int_{t_{0}}^{t_{1}}dw_{1,t}x_t+\int_{t_{0}}^{t_{1}}dw_{1,t}dx_{t}+\int_{t_{0}}^{t_{1}}d\beta_t b_t+\int_{t_{0}}^{t_{1}}d\beta_t db_t=0,
\end{equation}
to obtain the return on the rebalanced portfolio:
\begin{equation}\label{dptf}
\int_{t_{0}}^{t_{1}}dw_{t}=\int_{t_{0}}^{t_{1}}w_{1,t}dx_{t}+\int_{t_{0}}^{t_{1}}\beta_tdb_{t}.
\end{equation}
By applying It\^{o}'s lemma to $w=f\left(x,t\right)$, Black and Scholes obtain the option return:
\begin{equation}\label{doption}
\int_{t_{0}}^{t_{1}}dw_{t}=\int_{t_{0}}^{t_{1}}w_{1,t}dx_{t}+\int_{t_{0}}^{t_{1}}w_{2,t}dt+\frac{1}{2}\int_{t_{0}}^{t_{1}}w_{11,t}dx_t^2.
\end{equation}
Since Black and Scholes argue that the rebalanced portfolio replicates the option, equations (\ref{dptf}) and (\ref{doption}) need to be equal for all paths of the option value, which together with $\int\beta_tdb_{t}=\int\beta_tb_{t}rdt=\int w_{t}rdt-w_{1,t}x_{t}rdt$ yields:
\begin{equation}\label{integralequation}
\int_{t_{0}}^{t_{1}}w_{2,t}dt=\int_{t_{0}}^{t_{1}}w_{t}rdt-\int_{t_{0}}^{t_{1}}w_{1,t}x_{t}rdt-\frac{1}{2}\int_{t_{0}}^{t_{1}}w_{11,t}dx_t^2.
\end{equation}
Black and Scholes argue that equation (\ref{integralequation}) holds for:
\begin{equation}\label{pde}
w_{2,t}=w_{t}r-w_{1,t}x_{t}r-\frac{1}{2}w_{11,t}x_t^2\sigma^2,
\end{equation}
which is the partial differential equation that underlies their option pricing formula. However, substituting this partial differential equation in equation (\ref{integralequation}) yields:
\begin{equation}\label{hiha}
\int_{t_{0}}^{t_{1}}w_{11,t}dx_t^2=\int_{t_{0}}^{t_{1}}w_{11,t}x_t^2\sigma^2dt,
\end{equation}
which holds almost surely but not for every path, see Shreve (2004). Hence, the partial differential equation implies that there are paths where the return on the rebalanced portfolio in equation (\ref{dptf}) is not equal to the return on the option in equation (\ref{doption}), which contradicts the fact that the rebalanced portfolio replicates the option. Or, put differently, if the portfolio in equation (\ref{ptf}) replicates the option, the partial differential equation implies that there are paths where the portfolio return is not equal to equation (\ref{dptf}), so that for these paths the continuous-time self-financing condition does not hold.
\end{proof}

The above three fatal flaws imply that several results in the continuous-time finance literature are incorrect, including the self-financing condition and budget equation of Merton (1971), the continuous trading model of Harrison and Pliska (1981, 1983), and the hedging argument and option pricing formula of Black and Scholes.\footnote{Black and Scholes derived the formula using a hedging argument, but their paper also contains an alternative derivation based on the CAPM. This CAPM derivation is not used in the literature anymore, and Mink and de Weert (2024)\nocite{minkdeweert2024} show that it contains mathematical mistakes as well.} As an alternative to this formula, the value of an option can be expressed as its expected payoff at maturity discounted by its expected return. The next remark illustrates that when the value of the stock evolves according to equation (\ref{dx}), this expression resembles the option pricing formula of Black and Scholes, but contains the percentage drift on the stock and the (time-varying) percentage drift on the option instead of the risk-free rate of return.

\begin{remark}
At time $t$, the expected payoff on an option at maturity date $T$ is:
\begin{align}\label{expp}
\E\left[w_{T}\vert\mathcal{F}_{t}\right]=\int_{k}^{\infty}(x-k)dF_{T}(x)=\text{exp}\left(\mu\left(T-t\right)\right)N\left(d^{*}_{1,t}\right)x_{t}-N\left(d^{*}_{2,t}\right)k,
\end{align}
where $F_{T}(x) =\Pr\left(x_T\leq x\right|x_t)$ is the cumulative distribution function of $x_T$ conditional on $\mathcal{F}_{t}$. Note that $d^{*}_{1,t}=\frac{1}{\sigma\sqrt{T-t}}\left(\ln\left(\frac{x_{t}}{k}\right)+\left(\mu+\frac{1}{2}\sigma^{2}\right)\left(T-t\right)\right)$ and $d^{*}_{2,t}=d^{*}_{1,t}-\sigma\sqrt{T-t}$ differ from $d_{1,t}$ and $d_{2,t}$ in the formula of Black and Scholes unless $\mu=r$. The option value $w_t$ is equal to the expected payoff at maturity times the discount factor $\text{exp}\left(-\mu_{w,t}\left(T-t\right)\right)$, where $\mu_{w,t}$ is the (time-varying) percentage drift of the option. If $\mu_{w,t}=\mu=r$ the resulting expression is the same as the option pricing formula of Black and Scholes.
\end{remark}

\subsection{Binomial option pricing}\label{binomial}
While the previous analysis implies that the continuous-time option pricing formula of Black and Scholes is incorrect, Cox, Ross and Rubinstein (1979)\nocite{coxrossrubinstein1979} argue that this formula is a special limiting case of their discrete-time binomial option pricing formula. The previous analysis also invalidates the discrete-time hedging argument, however, by showing that the self-financing condition is incorrect in discrete time as well. Moreover, the next remark shows that even if we would set this problem aside, the analysis of Cox, Ross and Rubinstein (1979) actually illustrates that there are many cases in discrete time where an option cannot be replicated with a rebalanced portfolio of stocks and risk-free bonds.

\begin{remark}
Cox, Ross and Rubinstein (1979) state explicitly that in a model with three-state or trinomial stock price movements, an option cannot be replicated with a rebalanced portfolio of stocks and risk-free bonds. The reason is that any choice of $\alpha_{t}$ stocks and $\beta_{t}$ risk-free bonds can replicate the option value in two states but not in the third. More specifically, the analysis of Cox, Ross and Rubinstein (1979) implies that options cannot be replicated in discrete time when stock price movements follow a multinomial distribution with $n>2$ states.
\end{remark}

The remark raises the question of why an option cannot be hedged in a multinomial stock price model while it can be hedged in a binomial model (and even though the terminal values of a trinomial tree, for example, also converge to the lognormal distribution). The remainder of this section shows, however, that from a mathematical perspective the binomial model of Cox, Ross and Rubinstein (1979) has nothing to do with hedging or the self-financing condition to begin with. The reason is that the binomial model implicitly assumes that stocks and options that are exposed to systematic risk are not exposed to idiosyncratic risk. As we establish, this assumption directly implies their difference equation, without using the discrete-time self-financing condition or a hedging argument.

To derive their discrete-time option pricing formula, Cox, Ross and Rubinstein (1979) assume that the stock price and the option price each follow a binomial process, and point out that ``three-state or trinomial stock price movements will not lead to an option pricing formula based solely on arbitrage considerations." An implicit restriction of the binomial process is that the asset is only exposed to systematic risk and not to idiosyncratic risk. Specifically, assuming binomial processes for the stock and the option implies that each of them is exposed to just a single risk factor. Otherwise, if the stock would be exposed to two risk factors, the number of potential outcomes for its price would at least be equal to four instead of two, and likewise for the option. Since Cox, Ross and Rubinstein (1979) assume that the expected returns on the stock and the option exceed the risk-free return, their single risk factor is priced and is therefore systematic. The next proposition shows that this assumption directly implies the difference equation of Cox, Ross and Rubinstein (1979), without using the discrete-time self-financing condition or a hedging argument.

\begin{proposition}\label{p:crrsr}
The discrete-time difference equation of Cox, Ross and Rubinstein (1979) is directly implied -- without using the self-financing condition or a hedging argument -- by their assumption that the stock and the option are only exposed to systematic risk.
\end{proposition}
\begin{proof}
Cox, Ross and Rubinstein (1979) start with a one-period binomial tree in discrete time. Using their notation, the stock price at the start of the period is $S$ and the option price at this time is $C$. At the end of the period, the stock price is either equal to $uS$ or to $dS$, where $u>r^{*}>d>0$ and $r^{*}=1+r$ (we added an $^*$  to distinguish their notation from our symbol $r$). Likewise the option price at the end of the period is either $C^{u}$ or $C^{d}$. As this binomial tree implies that the stock and the option are only exposed to systematic risk, their excess rates of return are equal to:
\begin{align}
\left(uS-S\right)/S-(r^{*}-1)&=\beta_{S}\left(\left(m^{u}-m\right)/m-(r^{*}-1)\right), \\
\left(dS-S\right)/S-(r^{*}-1)&=\beta_{S}\left(\left(m^{d}-m\right)/m-(r^{*}-1)\right), \\
\left(C^{u}-C\right)/C-(r^{*}-1)&=\beta_{C}\left(\left(m^{u}-m\right)/m-(r^{*}-1)\right), \\
\left(C^{d}-C\right)/C-(r^{*}-1)&=\beta_{C}\left(\left(m^{d}-m\right)/m-(r^{*}-1)\right),
\end{align}
where $\left(m^{u}-m\right)/m$ is the rate of return on the market portfolio when the market goes up, $\left(m^{d}-m\right)/m$ is the rate of return when it goes down, and $\beta_{S}$ and $\beta_{C}$ indicate the exposure of the stock and of the option to the market portfolio. Combining these equations yields:
\begin{align}
C^{u}-Cr^{*}&=\frac{C\beta_{C}}{S\beta_{S}}\left(uS-Sr^{*}\right),\label{perfcor}\\
C^{d}-Cr^{*}&=\frac{C\beta_{C}}{S\beta_{S}}\left(dS-Sr^{*}\right).\label{perfcor2}
\end{align}
Using one of these equations to substitute for $\beta_{C}/\beta_{S}$ in the other gives:
\begin{equation}
C=\left[\left(\frac{r^{*}-d}{u-d}\right)C^{u}+\left(\frac{u-r^{*}}{u-d}\right)C^{d}\right]/r^{*},
\end{equation}
which is the difference equation in equation (2) of Cox, Ross and Rubinstein (1979).
\end{proof}

The previous proposition shows that the difference equation of Cox, Ross and Rubinstein (1979) is directly implied by the assumption of their binomial model that the stock and the option are only exposed to systematic risk. From a mathematical perspective, their difference equation is produced by this assumption rather than the result of a hedging argument. As a consequence, despite the mathematical commonalities between the binomial option pricing formula and the option pricing formula of Black and Scholes, the continuous-time hedging argument of Black and Scholes cannot be interpreted as the limit of a discrete-time hedging argument by Cox, Ross and Rubinstein (1979).

\section{Conclusion}
The discrete-time and continuous-time self-financing conditions that are used in the finance literature are mathematically misspecified. This result invalidates seminal contributions to this literature, including the continuous-time budget equation of Merton (1971), the hedging argument and option pricing formula of Black and Scholes (1973)\nocite{blackscholes1973}, the continuous trading model of Harrison and Pliska (1981, 1983), and the binomial option pricing model of Cox, Ross and Rubinstein (1979). Moreover, Black and Scholes (1973), Merton (1973b, 1977) and Harrison and Pliska (1981, 1983) implicitly assume the replication result that was to be established, while Cox, Ross and Rubinstein (1979) implicitly assume their difference equation without using a hedging argument. Hence, there is no longer a basis for the claim that options are redundant securities and that hedging them with the underlying stock yields a risk-free portfolio. Still, the value of an option can be expressed as its expected payoff at maturity discounted by its expected return. This expression resembles the option pricing formula of Black and Scholes (1973), but contains the percentage drift on the stock and the (time-varying) percentage drift on the option instead of the risk-free rate of return.

\bibliographystyle{apalike}
\bibliography{20230522MinkDeWeertarXiv}

\setcounter{section}{1}

\section*{Appendix}
\begin{propositionA}\label{correctsfcmerton}
For a rebalanced portfolio of $\alpha_{t}$ stocks and $\beta_{t}$ risk-free bonds, the correct continuous-time self-financing condition is $\int_{t_{0}}^{t_{1}}d\alpha_{t}x_{t}+\int_{t_{0}}^{t_{1}}d\beta_t b_{t}=0$, and the (path-by-path defined) processes $\alpha_{t}$ and $\beta_{t}$ must be continuous with bounded variation.
\end{propositionA}
\begin{proof}
Following Merton (1971), for a portfolio of $\alpha_{t}$ stocks and $\beta_{t}$ risk-free bonds, rebalancing in discrete time without inflows or outflows of external funds implies that:
\begin{equation}\label{discretesfc2}
\left[\Delta \alpha_{t}x_{t}+\Delta \alpha_{t}\Delta x_{t}+\Delta\beta_{t} b_{t}+\Delta \beta_{t}\Delta b_{t}\vert \mathcal{F}_{t}\right]=0.
\end{equation}
Note that equation (\ref{discretesfc2}) must hold path-by-path, otherwise there would be an inflow or outflow of funds for some realizations of $\Delta \alpha_{t}$ and $\Delta \beta_{t}$. When defining $\Delta t=\left(T-t\right)/n$ with $n>0$, the corresponding equation when letting $n$ go to $\infty$ must still hold path-by-path, which implies that $\lim_{n \to \infty}\Delta \alpha_{t}$ and $\lim_{n \to \infty}\Delta \beta_{t}$ must have bounded variation. Because of this bounded variation $\lim_{n \to \infty} \Delta \alpha_{t}\Delta x_t=0$ and $\lim_{n \to \infty} \Delta \beta_{t} \Delta x_t=0$, so that:
\begin{equation}\label{limsfcexante2}
\lim_{n \to \infty}\left[\Delta \alpha_{t}x_{t}+\Delta\beta_{t} b_{t}\vert \mathcal{F}_{t}\right]=0,
\end{equation}
which in stochastic integral notation gives the continuous-time self-financing condition:
\begin{equation}
\int_{t_{0}}^{t_{1}}d\alpha_{t}x_{t}+ \int_{t_{0}}^{t_{1}}d\beta_t b_{t}=0,
\end{equation}
where the processes $\alpha_{t}$ and $\beta_{t}$ are continuous with bounded variation.
\end{proof}

\begin{remarkA}\label{r:pathwise}
Black and Scholes argue that an investor who starts with $w_{1,t_0}$ stocks can rebalance this position so that he buys (or sells) $w_{1,s}-w_{1,t_0}$ stocks between time $t_0$ and $s$. In discrete time, Black and Scholes formalize this rebalancing by letting the investor change his stock position over discrete intervals $\Delta t=\left(s-t_0\right)/n$, where $n>0$ is a finite integer. For an integer $i>0$, at time $t_i=t_{0}+i\Delta t$ the investor buys $\Delta w_{1,t_i}=w_{1,t_{i}+\Delta t}-w_{1,t_i}$ stocks, so that the number of stocks bought between $t_{0}$ and $s$ is:
\begin{equation}\label{discrete}
w_{1,s}-w_{1,t_0}=\sum_{i=0}^{i=n-1}\Delta w_{1,t_i}.
\end{equation}
Because this equality holds path-by-path, it holds conditional on $\mathcal{F}_{t_{0}}$ and conditional on $\mathcal{F}_{s}$. At time $s$, the investor therefore knows exactly how many stocks he bought between $t_0$ and $s$, namely $\sum_{i=0}^{i=n-1}\Delta w_{1,t_i}$, so that this sum is well-defined and deterministic. Next, to formalize rebalancing in continuous time, Black and Scholes let $n \to \infty$ and argue that: 
\begin{equation}\label{continuous}
w_{1,s}-w_{1,t_0}=\lim_{n\to \infty} \sum_{i=0}^{i=n-1}\Delta w_{1,t_i}=\int_{t=t_{0}}^{t=s}dw_{1,t}.
\end{equation}
Their argument overlooks, however, that the It\^{o} integral $\int_{t_0}^{s} dw_{1,t}$ is not defined conditional on $\mathcal{F}_{s}$. That is, because the Wiener process has unbounded total variation the It\^{o} integral is not defined path-by-path, but is only defined almost surely and conditional on $\mathcal{F}_{t_{0}}$. The formalization of continuous rebalancing in equation (\ref{continuous}) therefore implies that $w_{1,s}-w_{1,t_{0}}$ is not defined conditional on $\mathcal{F}_{s}$, so that at time $s$ the investor does not know how many stocks he bought since $t_{0}$ and what his current position is.
\end{remarkA}

\begin{propositionA}\label{p:rp}
The option pricing formula of Black and Scholes implies that there are paths where the continuous-time self-financing condition does not hold.
\end{propositionA}
\begin{proof}
The option pricing formula of Black and Scholes states that the value of a call option with maturity date $T>t$ and strike price $k>0$ is equal to:
\begin{equation}\label{bsformula}
w_{t}=N\left(d_{1,t}\right)x_{t}-\text{exp}\left(-r\left(T-t\right)\right)N\left(d_{2,t}\right)k,
\end{equation}
where $d_{1,t}=\frac{1}{\sigma\sqrt{T-t}}\left(\ln\left(\frac{x_{t}}{k}\right)+\left(r+\frac{1}{2}\sigma^{2}\right)\left(T-t\right)\right)$ and $d_{2,t}=d_{1,t}-\sigma\sqrt{T-t}$, and where $N\left(\cdot\right)$ is the standard normal cumulative distribution function. Since $w_{1,t}=\partial w_{t}/\partial x_{t} =N\left(d_{1,t}\right)$, equation (\ref{bsformula}) can also be written as:
\begin{equation}\label{wval2}
w_{t}=w_{1,t}x_{t}+\beta_{t}b_{t},
\end{equation}
where $\beta_{t}=-N\left(d_{2,t}\right)$ and $b_{t}=\text{exp}\left(-r\left(T-t\right)\right)k$. Moreover, following Black and Scholes, using It\^{o}'s lemma implies that the option return equals:
\begin{equation}\label{dwte}
\int_{t_{0}}^{t_{1}}dw_t=\int_{t_{0}}^{t_{1}}w_{1,t}dx_t+\int_{t_{0}}^{t_{1}}w_{2,t}dt+\frac{1}{2}\int_{t_{0}}^{t_{1}}w_{11,t}dx_t^2,
\end{equation}
where $w_{2,t}=\partial w_{t}/\partial t$ is the partial derivative of the value of the option with respect to the remaining maturity of the option (i.e., theta), and where $w_{11,t}=\partial^{2}w_{t}/\partial x_{t}^{2}$ is the partial derivative of $w_{1,t}$ with respect to the value of the stock (i.e., gamma). As is well-known, equation (\ref{bsformula}) implies that these partial derivatives are equal to $w_{2,t}=-\frac{x_{t}N'\left(d_{1,t}\right)\sigma}{2\sqrt{T-t}}+\beta_{t}b_{t}r$ and $w_{11,t}=\frac{N'\left(d_{1,t}\right)}{x_{t}\sigma\sqrt{T-t}}$, where $N'\left(\cdot\right)$ is the standard normal probability density function. Substituting these partial derivatives in equation (\ref{dwte}) and using $b_{t}rdt=db_{t}$ yields:
\begin{align}\label{dw2}
\int_{t_{0}}^{t_{1}}dw_{t}&=\int_{t_{0}}^{t_{1}}w_{1,t}dx_{t}+\int_{t_{0}}^{t_{1}}\beta_{t}db_{t}+\frac{1}{2}\int_{t_{0}}^{t_{1}}w_{11,t}dx_t^2-\frac{1}{2}\int_{t_{0}}^{t_{1}}w_{11,t}\sigma^{2}x_t^2, \notag \\ 
&=\int_{t_{0}}^{t_{1}}w_{1,t}dx_{t}+\int_{t_{0}}^{t_{1}}\beta_{t}db_{t}.
\end{align}
Equations (\ref{wval2}) and (\ref{dw2}) are the same as equations (\ref{wval}) and (\ref{dw}), and therefore only hold if the continuous-time self-financing condition in equation (\ref{ctsfc}) holds. However, equation (\ref{dw2}) uses the equality $\int dx_{t}^{2}=\int\sigma^{2}x_{t}^{2}dt$, which holds almost surely but not pathwise (see Shreve, 2004). Hence, there are paths for which the second line of equation (\ref{dw2}) does not describe the return on the rebalanced portfolio, so that for these paths the continuous-time self-financing condition does not hold.
\end{proof}

\end{document}